\documentclass{llncs}

\usepackage{microtype,comment}
\usepackage{amssymb,graphicx,color,tikz,tkz-graph,amsmath,graphics}
\usetikzlibrary{calc}
\newcommand{\NP}{{\sf NP}}
\newcommand{\FPT}{{\sf FPT}}

\newcommand{\dist}{{\rm dist}}
\newcommand{\diam}{{\rm diam}}
\newcommand{\tw}{{\mathbf{tw}}}

\newcommand{\ad}{{\rm ad}}
\newcommand{\mad}{{\rm mad}}

\begin{document}
\title{ A Linear Kernel for Finding\\ Square Roots of Almost Planar Graphs
\thanks{This paper received support from  EPSRC (EP/G043434/1),  ANR project GraphEn 
(ANR-15-CE40-0009), the European Research Council under the European Union's Seventh Framework Programme (FP/2007-2013) / ERC Grant Agreement n. 267959 and the Research Council of Norway (the project CLASSIS).
An extended abstract of it appeared in the proceedings of SWAT 2016~\cite{GKPS16}.}}
\author{Petr A. Golovach\inst{1} \and Dieter Kratsch\inst{2} \and Dani{\"e}l Paulusma\inst{3} \and Anthony Stewart\inst{3}}
\institute{
Department of Informatics, 
University of Bergen, PB 7803, 5020 Bergen, Norway,\\
\texttt{petr.golovach@ii.uib.no}
\and
Laboratoire d'Informatique Th\'eorique et Appliqu\'ee, 
Universit\'e de Lorraine, 57045 Metz Cedex 01, France, \texttt{dieter.kratsch@univ-lorraine.fr}
\and
School of Engineering and  Computing Sciences, Durham University,\\ Durham DH1 3LE, UK,
\texttt{\{daniel.paulusma,a.g.stewart\}@durham.ac.uk}
}
\maketitle

\begin{abstract}
A graph $H$ is a square root of a graph $G$ if $G$ can be obtained from $H$ by the addition of edges between any two vertices in $H$ that are of distance~2 from each other.
The {\sc Square Root} problem is that of deciding whether a given graph admits a square root. We consider this problem for planar graphs in the context of 
the ``distance from triviality'' framework.
For an integer~$k$, a planar$+kv$ graph  
(or $k$-apex graph)
is a graph that can be made planar by the removal of at most~$k$ vertices.
We prove that a generalization of {\sc Square Root}, in which some edges are prescribed to be either in or out of any solution, has a
 kernel of size $O(k)$ for planar$+kv$ graphs, when parameterized by $k$.
Our result is based on a new edge reduction rule which, as we shall also show, has  a wider applicability for the {\sc Square Root} problem.
\end{abstract}

\section{Introduction}\label{s-intro}

Squares and square roots are well-known concepts in graph theory with a long history.
The {\it square} $G=H^2$ of a
graph $H=(V_H,E_H)$ is the graph with vertex set $V_G=V_H$, such that any two distinct vertices 
$u,v\in V_H$ are adjacent in $G$ if and only if $u$ and $v$ are of distance at most 2 in $H$.
A graph $H$ is a {\it square root} of $G$ if $G=H^2$. It is easy to check that there exist graphs with no square root,
graphs with a unique square root as well as graphs with many square roots. The corresponding recognition problem, which asks whether a given graph admits 
a square root, is called the {\sc Square Root} problem. Motwani and Sudan~\cite{MotwaniS94} showed that {\sc  Square Root} is \NP-complete.

\subsection{Existing Results}\label{s-existing}

In 1967, Mukhopadhyay~\cite{Mukhopadhyay67} characterized the graphs that have a square root. 
In line with the aforementioned \NP-completeness result of Motwani and Sudan, which appeared in 1994, this characterization does not lead to a polynomial-time
algorithm for {\sc Square Root}. Later results focussed on the following  
two recognition questions (${\cal G}$ denotes some fixed graph class):

\begin{itemize}
\item[(1)] How hard is it to recognize squares of graphs of~$\cal G$?
\item[(2)] How hard is it to recognize  graphs of $\cal G$ that have a square root?
\end{itemize}

\noindent
Note that the second question corresponds to the {\sc Square Root} problem restricted to graphs in ${\cal G}$, whereas the first question is the same as asking whether a given graph has a square root in ${\cal G}$.

Ross and Harary~\cite{RossH60} characterized squares of a tree and
proved that if a connected graph has a 
tree square root, then this root is unique up to isomorphism.
Lin and  Skiena~\cite{LinS95} gave a linear-time algorithm for recognizing squares of trees; they also proved that
{\sc Square Root} can be solved in linear time for planar graphs.
Le and Tuy~\cite{LeT10} generalized the above results for trees~\cite{LinS95,RossH60} to block graphs, 
whereas we recently gave a polynomial-time algorithm for recognizing squares of cactuses~\cite{GKPS16b}.
Nestoridis and Thilikos~\cite{NT14} proved that {\sc Square Root} is not only polynomial-time solvable for the class of planar graphs but for any non-trivial minor-closed graph class, that is, 
for any graph class that does not contain all graphs and that is closed under taking vertex deletions, edge deletions and edge contractions.

Lau~\cite{Lau06} gave a polynomial-time algorithm for recognizing squares of bipartite graphs; 
note that {\sc Square Root} is trivial for bipartite graphs, and even for $K_4$-free graphs, or equivalently, graphs of clique number at most~3, as
square roots of $K_4$-free graphs must have maximum degree at most~2.
Milanic, Oversberg and Schaudt~\cite{MOS14} proved that line graphs can only have bipartite graphs as a square root. The same authors also gave 
a linear-time algorithm for {\sc Square Root} restricted to line graphs.

Lau and Corneil~\cite{LauC04} gave a polynomial-time algorithm for 
recognizing squares of proper interval graphs and showed that the problems of recognizing squares of chordal graphs and squares of split graphs are both \NP-complete.
The same authors also proved that {\sc Square Root} is \NP-complete even for chordal graphs.
Le and Tuy~\cite{LeT11} gave a quadratic-time algorithm for recognizing squares of strongly chordal split graphs. 
 Le, Oversberg and Schaudt~\cite{LOS15} gave polynomial algorithms for recognizing squares of ptolemaic graphs and 3-sun-free split graphs.
In a more recent paper~\cite{LOS}, the same authors extended the latter result by giving polynomial-time results for recognizing squares of a number of other subclasses of split graphs.  
  Milanic and Schaudt~\cite{MilanicS13} proved that {\sc Square Root} can be solved in linear time for trivially perfect graphs and threshold graphs.
They posed the complexity of {\sc Square Root} restricted to split graphs and cographs as open problems.
Recently, we proved that {\sc Square Root} is linear-time solvable for 3-degenerate graphs and for $(K_r,P_t)$-free graphs for any two positive integers $r$ and $t$~\cite{GKPS}. 

Adamaszek and Adamaszek ~\cite{AdamaszekA11} proved that if a graph has a square root of girth at least 6, then this square root is unique up to isomorphism. Farzad, Lau, Le and Tuy~\cite{FarzadLLT12} showed that recognizing graphs with a square root of girth at least $g$ is polynomial-time solvable if $g\geq 6$ and \NP-complete if $g=4$. The missing case~$g=5$ was shown
to be \NP-complete by Farzad and Karimi~\cite{FarzadK12}.

In a previous paper~\cite{CochefertCGKP13} we proved that {\sc Square Root} is polynomial-time solvable for graphs of maximum degree~6.
We also considered square roots under the framework of parameterized complexity~\cite{CochefertCGKP13,CCGKP}.
We proved that the following two problems are fixed-parameter tractable with parameter $k$:
testing whether a connected $n$-vertex graph with $m$ edges has a square root with at most $n-1+k$ edges and testing
whether such a graph has a square root with at least $m-k$ edges. In particular, the first result implies that 
the problem of recognizing
squares of tree$+ke$ graphs, that is, 
graphs that can be modified into trees by removing at most $k$ edges,
is fixed-parameter tractable when parameterized by~$k$.

\subsection{Our Results}\label{s-our}

We are interested in developing techniques that  lead to new polynomial-time or parameterized algorithms for {\sc Square Root} for special graph classes. In particular, there are currently very few results on the parameterized complexity of {\sc Square Root}, and this is the main focus of our paper.

The graph classes that we consider fall under the ``distance from triviality'' framework, introduced by Guo, H\"uffner and Niedermeier~\cite{GHN04}.
For a graph class ${\cal G}$ and an 
integer~$k$ we define four classes of ``almost ${\cal G}$'' graphs, that is, graphs that are editing distance~$k$ apart from~${\cal G}$. To be more precise, the classes ${\cal G}+ke$, ${\cal G}-ke$, ${\cal G}+kv$ and ${\cal G}-kv$ consist of all
graphs that can be modified into a graph of ${\cal G}$ by deleting at most $k$ edges, adding at most $k$ edges, deleting at most $k$ vertices and adding at most $k$ vertices, respectively. Taking $k$ as the natural parameter, these graph classes have been well studied from a parameterized point of view for a number of problems. In particular this is true for the vertex coloring problem
restricted to (subclasses of) almost perfect graphs (due to the result of Gr\"otschel, Lov\'asz, and Schrijver~\cite{GLS84}, who proved that vertex coloring is polynomial-time solvable on perfect graphs). 

We consider ${\cal G}$ to be the class of {\it planar graphs}. 
As planar graphs are closed under taking edge and vertex deletions, the classes of planar$-kv$ graphs and planar$-ke$ graphs coincide with planar graphs.
Hence, we only need to consider planar$+kv$ graphs and planar$+ke$ graphs, that is, graphs that can be made planar by at most~$k$ vertex deletions or at most~$k$ edge deletions, respectively. 
We note that planar$+kv$ graphs are also known as {\it $k$-apex graphs}.
Moreover, we observe that {\sc Square Root} is \NP-complete for planar$+kv$ graphs and planar$+ke$ graphs when $k$ is part of the input, as the classes of planar$+nv$ graphs and planar$+n^2e$ graphs coincide with the class of all graphs on $n$ vertices.

Our main contribution is showing  
that {\sc Square Root} is \FPT{} on $k$-apex graphs when parameterized by $k$. More precisely, we prove that a more general version of the problem admits a linear kernel. The \textsc{Square Root with Labels} problem takes as input a graph $G$ with two subsets $R$ and $B$ of prespecified edges: the edges of $R$ need to be included in a solution (square root) and the edges of $B$
 are forbidden in the solution.
We prove that {\sc Square Root with Labels} has a kernel of size $O(k)$ for planar$+kv$ graphs, when parameterized by~$k$. 
As every $planar+ke$ graph is $planar+kv$, we immediately obtain
 the same result for planar$+ke$ graphs.
The {\sc Square Root with Labels} problem was
introduced 
in a previous paper~\cite{CochefertCGKP13}, but in this paper we introduce a new reduction rule, which we call the {\it edge reduction rule}.

The edge reduction rule is used to recognize, in polynomial time, a certain local substructure that graphs with square roots must have. As such, our rule can be added to the list of known and similar polynomial-time reduction rules for recognizing square roots.
To give a few examples, the reduction rule of Lin and Skiena~\cite{LinS95} is based on recognizing pendant edges and bridges of square roots of planar graphs,
whereas the reduction rule of Farzad, Le and Tuy~\cite{FarzadLLT12} is based on the fact that squares of graphs with large girth can be recognized to have a unique root.
In contrast, our edge reduction rule, which is based on detecting so-called recognizable edges whose neighbourhoods have some special property
(see Section~\ref{s-edge} for a formal description) is tailored for graphs with no unique square root, just as we did  in~\cite{CCGKP}; in fact our new rule, which we explain in detail in Section~\ref{s-reduce}, can be seen as an improved and more powerful variant of the rule used in~\cite{CCGKP}.
For squares with no unique square root, not all the root edges can be recognized in polynomial time. Hence, removing certain local substructures, thereby reducing the graph to a smaller graph, and keeping track of the compulsory edges (the recognized edges) and forbidden edges is the best we can do.  
However, after the reduction, the 
connected components of the remaining graph might be dealt with further by exploiting the properties of the graph class under consideration.
This is exactly what we do for planar$+kv$ graphs to obtain the linear kernel in Section~\ref{s-kernel}.

The fact that our edge reduction rule is more general than the other known rules is also evidenced by other applications of it. 
In~\cite{CCGKPS} we showed that it can be used to obtain an alternative proof of the known
result~\cite{CochefertCGKP13} that {\sc Square Root} is polynomial-time solvable for graphs of maximum degree at most~6.\footnote{The proof in~\cite{CochefertCGKP13} is based on a different and less general reduction rule, which only ensures boundedness of treewidth, while the edge reduction rule yields graphs of maximum degree at most~6 with a bounded number of vertices.} 
As a third application of our edge reduction rule we show in Section~\ref{sec:mad} that it can be used to solve {\sc Square Root} in
polynomial-time solvable for graphs of maximum average degree smaller than $\frac{46}{11}$. 

In Section~\ref{s-con} we give some directions for future work.

\section{Preliminaries}\label{sec:defs}
We only consider finite undirected graphs without loops or multiple edges. 
We refer to the textbook by Diestel~\cite{Diestel10} for any undefined graph terminology.

We denote the vertex set of a graph $G$ by $V_G$ and the edge set by $E_G$. The subgraph of $G$
induced by a subset $U\subseteq V_G$ is denoted by $G[U]$. 
The graph $G-U$ is the graph obtained from $G$ after removing the vertices of $U$. If $U=\{u\}$, we also write $G-u$. 
Similarly, we denote the graph obtained from $G$ after deleting an edge $e$ by $G-e$.
A vertex $u$ is a \emph{cut vertex} of a connected graph $G$ with at least two vertices if $G-u$ is disconnected.
An inclusion-maximal subgraph of $G$ that has no cut vertices is called a \emph{block}. 
A {\it bridge} of a connected graph $G$ is an edge $e$ such that $G-e$ is disconnected.

In the remainder of this section let $G$ be a graph.
We say that $G$ is planar$+kv$ if $G$ can be made planar 
by removing at most~$k$ vertices.
The \emph{distance} $\dist_G(u,v)$ between a pair of vertices $u$ and $v$ of~$G$ is the number of edges of a shortest path between them. 
The diameter $\diam(G)$ of~$G$ is the maximum distance between any two vertices of $G$. 
The distance between a vertex  $u\in V_G$ and a subset $X\subseteq V_G$ is denoted by $\dist_G(u,X)=\min\{\dist_G(u,v)\mid v\in X\}$.
The distance between two subsets $X$ and $Y$ of $V_G$ is denoted by $\dist_G(X,Y)=\min\{\dist_G(u,v)\mid u\in X,v\in Y\}$.
Whenever we speak about the distance between a vertex set $X$ and a subgraph $H$ of $G$, we mean the distance between $X$ and $V_H$.

The \emph{open neighbourhood} of a vertex $u\in V_G$ is defined as $N_G(u) = \{v\; |\; uv\in E_G\}$ and
its \emph{closed neighbourhood} is defined as $N_G[u] = N_G(u) \cup \{u\}$. 
For $X\subseteq V_G$, let $N_G(X)=\bigcup_{u\in X}N_G(u)\setminus X$.
Two (adjacent) vertices $u,v$ are said to be \emph{true twins} if $N_G[u]=N_G[v]$.
The degree of a vertex
$u\in V_G$ is defined as $d_G(u)=|N_G(u)|$.
The maximum degree of $G$ is $\Delta(G)=\max\{d_G(v)\; |\; v\in V_G\}$.
A vertex of degree~1 is said to be a \emph{pendant} vertex. If $v$ is a pendant vertex, then we say the unique edge incident to $u$ is a \emph{pendant} edge.

The framework of parameterized complexity allows us to study the computational complexity of a discrete optimization problem in two dimensions.
One dimension is the input size~$n$ and the other one is a parameter $k$. We refer to the recent textbook of Cygan et al.~\cite{CyganFKLMPPS15} for further details and only give the definitions for those notions relevant for our paper here.
A parameterized problem is \emph{fixed parameter tractable} (\FPT) if it can be solved in time $f(k)\cdot n^{O(1)}$ for some computable function $f$. 
A \emph{kernelization} of a parameterized problem~$\Pi$ is a polynomial-time algorithm that maps each instance $(x,k)$ with input $x$ and parameter $k$ to an instance $(x',k')$, such that i) $(x,k)$ is a yes-instance if and only if $(x',k')$ is a yes-instance of $\Pi$, and ii) $|x'|+k'$ is bounded by $f(k)$ for some computable function $f$. 
The output $(x',k')$ is called a \emph{kernel} for $\Pi$. 
The function~$f$ is said to be a \emph{size} of the kernel.
It is well known that a decidable parameterized problem is \FPT{} if and only if it has a kernel. 
A logical next step is then to try to reduce the size of the kernel.
We say that $(x',k')$ is a \emph{linear} kernel if $f$ is linear.

\section{Recognizable Edges}\label{s-edge}

In this section we introduce the definition of a recognizable edge, which plays a crucial role in our paper, together with the corresponding notion of a $(u,v)$-partition. We also prove some important lemmas about this type of edges. See Fig.~\ref{fig:recog} (i) for an example of a recognizable edge and a corresponding $(u,v)$-partition $(X,Y)$.

\begin{definition}\label{d-rec}
An edge $uv$ of a graph $G$ is said to be \emph{recognizable} if the following four conditions are satisfied:
\begin{itemize}
\item[a)] $N_G(u)\cap N_G(v)$ has a partition $(X,Y)$ where 
$X=\{x_1,\ldots,x_p\}$ and $Y=\{y_1,\ldots,y_q\}$, $p,q\geq 1$, are 
(disjoint) cliques in $G$;
\item[b)] $x_iy_j\notin E_G$ for $i\in\{1,\ldots,p\}$ and $j\in\{1,\ldots,q\}$;
\item[c)] for any $w\in N_G(u)\setminus N_G[v]$, $wy_j\notin E_G$ for  $j\in\{1,\ldots,q\}$, and symmetrically,  for any $w\in N_G(v)\setminus N_G[u]$, $wx_i\notin E_G$ for  $i\in\{1,\ldots,p\}$;
\item[d)] for any $w\in N_G(u)\setminus N_G[v]$, there is an $i\in\{1,\ldots,p\}$ such that 
$wx_i\in E_G$, and symmetrically,  for any $w\in N_G(v)\setminus N_G[u]$, there is a $j\in\{1,\ldots,q\}$
such that $wy_j\in E_G$.
\end{itemize}
We also call such a partition $(X,Y)$ a \emph{$(u,v)$-partition} of $N_G(u)\cap N_G(v)$. 
\end{definition}

\noindent
Notice that due to c) and d), $(X,Y)$ is an ordered pair defined for an ordered pair $(u,v)$; if $N_G(u)\setminus N_G(v)\neq\emptyset$ or 
$N_G(v)\setminus N_G(u)\neq\emptyset$ 
then  $(Y,X)$ is not a \emph{$(u,v)$-partition},
as condition c) is violated 
(and in some instances, condition d) as well).

\begin{figure}[ht]
\centering\scalebox{0.9}{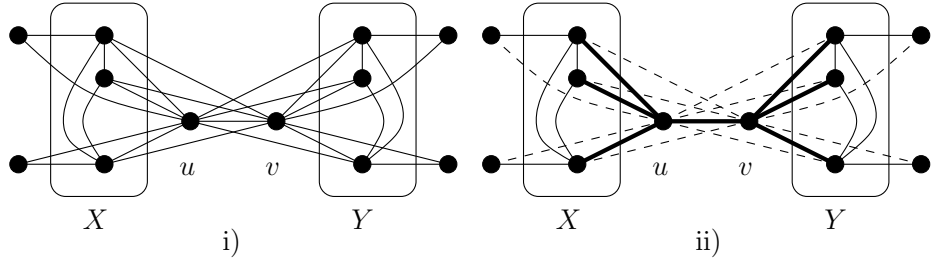}
\caption{(i) An example of a graph $G$ with a recognizable edge
$uv$ and a corresponding $(u,v)$-partition $(X,Y)$. 
(ii) A square root of $G$. In this figure,
the edges of the square root are shown by thick lines and the edges of $G$ not belonging to the square root  
are shown by dashed lines. 
Edges which may or may not belong to the square root are shown by neither thick nor dashed lines.
\label{fig:recog}}
\end{figure}

In the next lemma we give a necessary condition of an edge of a square root $H$ of a graph~$G$ to be recognizable in $G$.
In particular, this lemma implies that any non-pendant bridge of $H$ is a recognizable edge of $G$.

\begin{lemma}\label{lem:edge-one}
Let $H$ be a square root of  a graph $G$. Let $uv$ be an edge of $H$ that is not pendant and such that
any cycle in $H$ containing $uv$ has length at least~$7$. 
Then $uv$ is a recognizable edge of $G$ 
and $(N_H(u)\setminus\{v\},N_H(v)\setminus\{u\})$ is a $(u,v)$-partition in $G$.
\end{lemma}

\begin{proof}
Let $H$ be a square root of a graph $G$ and let $uv$ be an edge of $H$ such that 
$uv$ is not a pendant edge of $H$ and any cycle in 
$H$ containing $uv$ has length at least $7$. Let $X=\{x_1,\ldots,x_p\}=N_H(u)\setminus\{v\}$ and $Y=\{y_1,\ldots,y_q\}=N_H(v)\setminus\{u\}$.
Because $uv$ is not a pendant edge and any cycle in $H$ that contains $uv$ has length at least 7, it follows that $X\neq\emptyset$, $Y\neq\emptyset$ and $X\cap Y=\emptyset$.
We show that $(X,Y)$ is a $(u,v)$-partition of $N_G(u)\cap N_G(v)$  in $G$ by proving that conditions a)--d) of 
Definition~\ref{d-rec}
are fulfilled.

First we prove a). Let $z\in N_G(u)\cap N_G(v)$. We will show that $z\in X\cup Y$. 
If $uz\in E_H$ then $z\in X$, and if $vz\in E_H$ then $z\in Y$.
Suppose that $z\notin X$ and $z\notin Y$.
Since $uz\in E_G$, there is a vertex~$w\in V_G$ such that $uw,wz\in E_H$.
Since $vz\notin E_H$ it follows that $w\neq v$.
It follows due to symmetry that there exists $w'\in V_G$ such that $vw',w'z\in E_H$ and $w'\neq u$.
Then either $wuvw'$ is a cycle in $H$ if $w=w'$, otherwise, $zwuvw'z$ is a cycle of $H$.
In both cases we have a contradiction since any cycle in 
$H$ containing $uv$ has length at least $7$. This proves that $z\in X\cup Y$ and therefore, $N_G(u)\cap N_G(v)\subseteq X\cup Y$. 
Since $vx_i\in E_G$ and $uy_j\in E_G$ for all $i\in \{1,\ldots,p\}$ and $j\in\{1,\ldots,q\}$, we see that $X\cup Y\subseteq N_G(u)\cap N_G(v)$.
Because $X,Y\neq\emptyset$ and $X\cap Y=\emptyset$, $(X,Y)$ is a partition of $N_G(u)\cup N_G(v)$.
It remains to observe that $X$ and $Y$ are cliques in $G$ because any two vertices of~$X$ and any two vertices of $Y$ have $u$ or $v$, respectively, as common neighbour in $H$.  

To prove b), assume that there are $i\in\{1,\ldots,p\}$ and $j\in\{1,\ldots,q\}$ such that $x_iy_j\in E_G$. Because $H$ has no cycle of length 4 containing $uv$,
 $x_iy_j\notin E_H$.
Hence, there is $z\in V_H$ such that $x_iz,zy_j\in E_H$. Because $H$ has no cycles of length~3 containing $uv$, we find that $z\notin\{u,v\}$. We conclude that $zx_iuvy_jz$ is a cycle of length~5 in~$H$ that contains $uv$; a contradiction.

To prove c), it suffices to show that for any $w\in N_G(u)\setminus N_G[v]$, $wy_j\notin E_G$ for  $j\in\{1,\ldots,q\}$, as the second part is symmetric.
To obtain a contradiction, assume that there are vertices $w\in N_G(u)\setminus N_G[v]$ and $y_j$ for some $j\in\{1,\ldots,q\}$ such that $wy_j\in E_G$.
By a), $(X,Y)$ is a partition of $N_G(u)\cap N_G(v)$. Hence, $w\notin X$ and $w\notin Y$. Because $w\notin X$ and $w\in N_G(u)$, there is $x\in V_G$ such that $ux,xw\in E_H$. As $ux\in E_H$, we have $x\in X$.  If $wy_j\in E_H$, then the cycle $uxwy_jvu$ containing $uv$ has  length~5; a contradiction.
Hence, $wy_j\notin E_H$. Because $wy_j\in E_G$, there is a vertex $z\in V_H$ such that $wz,zy_j\in E_H$. Since $w\in N_G(u)\setminus N_G[v]$, we have  $w\notin\{u,v\}$. If $x=z$, then $uvy_jxu$ is a cycle of length~4 containing  $uv$, a contradiction. If $x\neq z$, then $uvy_jzwxu$ is a cycle of length~6 containing $uv$, another contradiction.

To prove d) we consider some $w \in N_G(u)\setminus N_G[v]$. We note that since $X\subseteq N_G(u)\cap N_G(v)$, $w\notin X$ and thus $uw\notin E_H$.
Since $uw \in  E_G$ by definition, there must be some $x\in V_G$ such that $ux,xw \in E_H$.
Because $w$ is not adjacent to $v$, we find that $x\neq v$. 
Since $ux \in E_H$ and $X=N_H(u)\setminus\{v\}$, this means that $x\in X$.
The second condition in d) follows by symmetry.\qed
\end{proof}

The following corollary follows immediately from Lemma~\ref{lem:edge-one}.

\begin{corollary}\label{lem:edge-one2}
Let $H$ be a square root of a graph with no recognizable edges. Then every non-pendant edge of $H$ lies on a cycle of length at
most~$6$.
\end{corollary}

In Lemma~\ref{lem:edge-two} we show that recognizable edges in a graph $G$ can be used to identify some edges of a square root of $G$ and also some edges that are not included in any square root of $G$;
see Fig.~\ref{fig:recog} (ii) for an illustration of this lemma.

\begin{lemma}\label{lem:edge-two}
Let $G$ be a graph with a square root $H$. Additionally let $uv$ be a recognizable edge of $G$ with a 
$(u,v)$-partition $(X,Y)$ where $X=\{x_1,\ldots,x_p\}$ and $Y=\{y_1,\ldots,y_q\}$.
Then: 
\begin{itemize}
\item[i)] $uv\in E_H$;
\item[ii)] for every $w\in N_G(u)\setminus N_G[v]$, $wu \notin E_H$, and
for every $w\in N_G(v)\setminus N_G[u]$, $wv \notin E_H$.
\item[iii)] if $u,v$ are true twins in $G$, then either 
$ux_1,\ldots,ux_p\in E_H$,  $vy_1,\ldots,vy_q\in E_H$ and 
$uy_1,\ldots,uy_q\notin E_H$,  $vx_1,\ldots,vx_p\notin E_H$ or
$ux_1,\ldots,ux_p\notin E_H$,  $vy_1,\ldots,vy_q\notin E_H$ and 
$uy_1,\ldots,uy_q\in E_H$,  $vx_1,\ldots,vx_p\in E_H$;
\item[iv)] if $u,v$ are not true twins in $G$, then
$ux_1,\ldots,ux_p\in E_H$,  $vy_1,\ldots,vy_q\in E_H$ and 
$uy_1,\ldots,uy_q\notin E_H$,  $vx_1,\ldots,vx_p\notin E_H$.
\end{itemize}
\end{lemma}

\begin{proof} The proof
uses conditions a)--d) of
Definition~\ref{d-rec}.

To prove i), suppose that $uv\notin E_H$. Then there is a vertex $z\in N_G(u)\cap N_G(v)$ such that $zu,zv\in E_H$. Assume without loss of generality that $z\in X$. Because of b), $zy_1\notin E_G$, which implies, together with $zv\in E_H$, that $vy_1\notin E_H$. 
Because $vy_1\in E_G$, this means that there is a vertex~$w$ with 
 $vw,wy_1\in E_H$. Because we assume $uv\notin E_H$, we observe that $w\neq u$.
By b), $w\notin X$ and, therefore, $w\in N_G(v)\setminus N_G(u)$. As $zv,vw\in E_H$, we obtain $wz \in E_G$. However, as $z\in X$, this contradicts c). We conclude that $uv\in E_H$.

To prove ii), it suffices to consider the case in which $w\in N_G(u)\setminus N_G[v]$, as the other case is symmetric. If $wu \in E_H$, then because $uv\in E_H$, we have $wv\in E_G$ contradicting $w\notin N_G(v)$.

We now prove iii) and iv). First suppose 
that there exist vertices $x_i$ and $x_j$ (with possibly $i=j$) for some $i,j\in \{1,\ldots,p\}$ such that
$x_iu,x_jv\in E_H$. 
Then, as $x_iy_1, x_jy_1\notin E_G$ by b),
we find that $y_1u,y_1v\notin E_H$. 
As $y_1u\in E_G$, the fact that $y_1u\notin E_H$ means that there exists a vertex $w\in V_H\setminus \{u\}$
 such that $wu,wy_1\in E_H$. As $y_1v\notin E_H$, we find that $w\neq v$, so $w\in V_H\setminus \{u,v\}$.
 As $x_iu,uw\in E_H$, we find that $x_iw\in E_G$, consequently $w\notin Y$ due to b).
 Because $wy_1\in E_H$ we obtain $w\notin X$, again due to b). Hence, $w\notin X\cup Y=N_G(u)\cap N_G(v)$. Therefore, as $uw\in E_G$ and $w\neq v$, we have $w\in N_G(u)\setminus N_G[v]$, but as $wy_1\in E_G$ this contradicts c). 
Hence, this situation cannot occur.

Suppose that there a vertex $x_i$ for some $i\in \{1,\ldots,p\}$ such that $x_iu,x_iv\notin E_H$. Then, as $x_iv\in E_G$, there exists a vertex
$w\in V_H\setminus \{u,v\}$, such that $wv,wx_i\in E_H$. By b), $w\notin Y$. As $uv\in E_H$ due to statement i) and $vw\in E_H$, we find that $uw\in E_G$. Hence, as $w\notin Y$, we obtain $w\in X$. 
As $x_iu \in E_G\setminus E_H$ and $x_iv\notin E_H$, there is a vertex $z\in V_H\setminus \{u,v\}$ such that $zu,zx_i\in E_H$. As $uv\in E_H$ due to statement i), this implies that $zv\in E_G$. Hence, $z\in X\cup Y$. As $zx_i\in E_H$, we find that $z\notin Y$ due to b). Consequently, $z\in X$. This means that we have vertices $w,z\in X$ (possibly $w=z$) and edges  $zu, wv\in E_H$. However, we already proved above that this is not possible.
We obtain that either $ux_1,\ldots,ux_p\in E_H$ and $vx_1,\ldots,vx_p\notin E_H$, or that
$ux_1,\ldots,ux_p\notin E_H$ and $vx_1,\ldots,vx_p\in E_H$. Symmetrically,
either $uy_1,\ldots,uy_q\in E_H$ and $vy_1,\ldots,vy_q\notin E_H$, or
$uy_1,\ldots,uy_q\notin E_H$ and $vy_1,\ldots,vy_q\in E_H$. By b),  it cannot happen that 
$ux_1,uy_1\in E_H$ or $vx_1,vy_1\in E_H$. Hence, either 
$ux_1,\ldots,ux_p\in E_H$,  $vy_1,\ldots,vy_q\in E_H$ and 
$uy_1,\ldots,uy_q\notin E_H$,  $vx_1,\ldots,vx_p\notin E_H$ or
$ux_1,\ldots,ux_p\notin E_H$,  $vy_1,\ldots,vy_q\notin E_H$ and 
$uy_1,\ldots,uy_q\in E_H$,  $vx_1,\ldots,vx_p\in E_H$. 
In particular, this implies iii).

To prove iv), assume without loss of generality that $N_G(u)\setminus N_G[v]\neq \emptyset$. 
For contradiction, let $ux_1,\ldots,ux_p\notin E_H$,  $vy_1,\ldots,vy_q\notin E_H$ and 
$uy_1,\ldots,uy_q\in E_H$,  $vx_1,\ldots,vx_p\in E_H$. Let $w\in N_G(u)\setminus N_G[v]$. 
By d), there is a vertex $x_i$ for some $i\in\{1,\ldots,p\}$ such that $wx_i\in E_G$. 
Then $wx_i\notin E_H$, as otherwise our assumption that $vx_i\in E_H$ will imply that $w\in N_G(v)$, which is not possible. Since $wx_i\in E_G\setminus E_H$, there exists a vertex~$z\in V_H$, such that $zw,zx_i\in E_H$.
Because $x_iu\notin E_H$, we find that $z\neq u$, and because $w\notin N_G(v)$, we find that $z\neq v$.
Because $zx_i,x_iv\in E_H$, we obtain $zv\in E_G$. 
As $w\notin N_G(v)$ and $vx_j\in E_H$ for all $j\in \{1,\ldots,p\}$, we have $wx_j\notin E_H$ for all $j\in \{1,\ldots,p\}$.
Hence, as $zw\in E_H$, we find that $z\notin X$. As $zx_i\in E_H$, we find that $z\notin Y$ due to b). 
Hence, $z\notin X\cup Y=N_G(u)\cap N_G(v)$. As $zv\in E_G$, this implies that $z\in N_G(v)\setminus N_G[u]$ (recall that $z\neq u$). 
Because $zx_i\in E_G$, this is in contradiction with c).\qed
\end{proof}

\noindent
{\bf Remark 1.}
If the vertices $u$ and $v$ of the recognizable edge of the square~$G$ in  Lemma~\ref{lem:edge-two} are true twins, then by statement~iii) of this lemma and the fact that the vertices $u$ and $v$ are interchangeable, $G$ has 
at least two isomorphic square roots: one root containing $ux_1,\ldots,ux_p$,  $vy_1,\ldots,vy_q$ and excluding
$uy_1,\ldots,uy_q$,  $vx_1,\ldots,vx_p$, and another one containing 
$ux_1,\ldots,ux_p$,  $vy_1,\ldots,vy_q$ and excluding 
$uy_1,\ldots,uy_q$,  $vx_1,\ldots,vx_p$.

\section{The Edge Reduction Rule}\label{s-reduce}

In this section we present our edge reduction rule.
As mentioned in Section~\ref{s-our}, we solve a more general problem than {\sc Square Root}.
Before discussing the edge reduction rule, we first formally define this problem (see also~\cite{CochefertCGKP13}).

\begin{description}
\item [{\sc Square Root with Labels}] 
\item[Input:] a graph $G$ and two sets of edges $R,B\subseteq E_G$.
\item[Question:] is there a graph $H$ with $H^2=G$, 
$R\subseteq E_H$ and $B\cap E_H=\emptyset$?
\end{description}

\noindent
Note that {\sc Square Root}  is indeed a special case of {\sc Square Root with Labels}: choose $R=B=\emptyset$. 

We say that a graph $H$ is a \emph{solution} for an instance $(G,R,B)$ of {\sc Square Root with Labels} if $H$ satisfies the following three conditions:
(i) $H^2=G$; (ii) $R\subseteq E_H$; and (iii)~$B\cap E_H=\emptyset$.

\medskip
\noindent
We use Lemmas~\ref{lem:edge-one} and \ref{lem:edge-two} to preprocess instances of {\sc Square Root with Labels}. Our edge reduction algorithm takes as input an instance $(G,R,B)$ of {\sc Square Root with Labels} and either returns an equivalent instance 
with no recognizable edges or answers {\sc no}. 

\medskip
\noindent
{\bf Edge Reduction}

\begin{enumerate}
\item Find  a recognizable edge $uv$ together with corresponding $(u,v)$-partition $(X,Y)$, $X=\{x_1,\ldots,x_p\}$ and $Y=\{y_1,\ldots,y_q\}$. If such an edge $uv$ does not exist, then return the obtained instance of   {\sc Square Root with Labels} and stop. 

\item 
If $uv\in B$ then return {\tt no} and stop. Otherwise let $B_1=\{wu\; | \; w\in N_G(u)\setminus N_G[v]\}\cup \{wv\; |\; w\in N_G(v)\setminus N_G[u]\}$. If $R\cap B_1\neq\emptyset$, then return {\tt no} and stop.

\item
If $u$ and $v$ are not true twins then set $R_2=\{ux_1,\ldots,ux_p\} \cup \{vy_1,\ldots,vy_q\}$ and $B_2=\{uy_1,\ldots,uy_q\} \cup \{vx_1,\ldots,vx_p\}$. If $R_2\cap B\neq \emptyset$ or $B_2\cap R\neq\emptyset$, then return {\tt no} and stop.

\item
If $u$ and $v$ are true twins then do as follows:
\begin{enumerate}
\item If $(\{uy_1,\ldots,uy_q\} \cup \{vx_1,\ldots,vx_p\}) \cap R\neq \emptyset$ or\\ 
\hspace*{4mm}$(\{ux_1,\ldots,ux_p\} \cup \{vy_1,\ldots,vy_q\}) \cap B\neq \emptyset$ then\\
\hspace*{4mm}set $R_2=\{uy_1,\ldots,uy_q\} \cup \{vx_1,\ldots,vx_p\}$ and\\ \hspace*{9mm}$B_2=\{ux_1,\ldots,ux_p\} \cup \{vy_1,\ldots,vy_q\}$.\\
If $R_2\cap B\neq \emptyset$ or $B_2\cap R\neq\emptyset$, then return {\tt no} and stop.\\[-10pt]
\item If $(\{uy_1,\ldots,uy_q\} \cup \{vx_1,\ldots,vx_p\}) \cap R= \emptyset$ and\\
\hspace*{4mm}$(\{ux_1,\ldots,ux_p\} \cup \{vy_1,\ldots,vy_q\}) \cap B= \emptyset$ then\\
\hspace*{4mm}set $R_2=\{ux_1,\ldots,ux_p\} \cup \{vy_1,\ldots,vu_q\}$ and\\ \hspace*{9mm}$B_2=\{uy_1,\ldots,uy_q\} \cup \{vx_1,\ldots,vx_p\}$.\\
(Note that $R_2\cap B=\emptyset$ and $B_2\cap R=\emptyset$.)
\end{enumerate}
\item Delete the edge $uv$ and the edges of $B_1$ from $G$, set $R:=(R\setminus \{uv\})\cup R_2$ and $B:=(B\setminus B_1)\cup B_2$, and return to Step~1.
\end{enumerate}

\begin{lemma}\label{lem:preproc}
For an instance $(G,R,B)$ of {\sc Square Root with Labels} where $G$ has $n$ vertices and $m$ 
edges,  {\bf Edge Reduction} in time $O(n^2m^2)$ either correctly answers {\sc no} or returns an equivalent instance $(G',R',B')$  with the following property: for any square root $H$ of $G'$, every edge of $H$ is either a pendant edge of $H$ or is included in a cycle of length at most 6 in $H$. Moreover, $(G',R',B')$ has a solution~$H$ if and only if $(G,R,B)$ has a solution that can be obtained from $H$ by restoring all recognizable edges.
\end{lemma}

\begin{proof}
It suffices to consider one iteration of the algorithm to prove its correctness.
If we stop at Step~1 and return the obtained instance of {\sc Minimum Square Root with Labels}, then by Lemma~\ref{lem:edge-one}, for any square root $H$ of $G'$, every non-pendant edge of $H$ is included in a cycle of length at most 6 in $H$.

To show the correctness of Step~2, we note that
by Lemma~\ref{lem:edge-two}~i), $uv$ is included in any square root and the edges of $B_1$ are not included in any square root. Hence, if what we do in Step~2 is not consistent with $R$ and $B$, there is no square root of $G$ that includes the edges of $R$ and excludes the edges of $B$, thus returning output {\tt no} is correct.

To show the correctness of Step~3, suppose  $u$ and $v$ are not true twins. 
Then by Lemma~\ref{lem:edge-two}~iv) it follows that
$ux_1,\ldots,ux_p\in E_H$, $vy_1,\ldots,vy_q\in E_H$, $uy_1,\ldots,uy_q\notin E_H$ and $vx_1,\ldots,vx_p\notin E_H$ for any square root $H$.
Hence, we must define $R_2$ and $B_2$ according to this lemma. 
If afterwards we find that $R_2\cap B\neq \emptyset$ or $B_2\cap R\neq\emptyset$, then $R_2$ or $B_2$ is not consistent with $R$ or $B$, respectively, and thus, retuning {\tt no} if this case happens is correct.

To show the correctness of Step~4, suppose that $u$ and $v$ are true twins. Then by Lemma~\ref{lem:edge-two}~iv) we have two options.  First, if $(\{uy_1,\ldots,uy_q\} \cup \{vx_1,\ldots,vx_p\}) \cap R\neq \emptyset$ or  $(\{ux_1,\ldots,ux_p\} \cup \{vy_1,\ldots,vy_q\}) \cap B\neq \emptyset$, then we are forced to go for the option as defined in Step~4(a). If afterwards $R_2\cap B\neq \emptyset$ or $B_2\cap R\neq\emptyset$, then we still need to return {\tt no} as in Step~3. Second, if $(\{uy_1,\ldots,uy_q\} \cup \{vx_1,\ldots,vx_p\}) \cap R= \emptyset$ and $(\{ux_1,\ldots,ux_p\} \cup \{vy_1,\ldots,vy_q\}) \cap B= \emptyset$, then we may set without loss of generality  (cf.~Remark~1) that $R_2=\{ux_1,\ldots,ux_p\} \cup \{vy_1,\ldots,vu_q\}$ and $B_2=\{uy_1,\ldots,uy_q\} \cup \{vx_1,\ldots,vx_p\}$. Note that in this case  $R_2\cap B=\emptyset$ and $B_2\cap R=\emptyset$.

Finally, to show the correctness of Step~5, let $G'$ be the graph obtained from $G$ after deleting  
the edge $uv$ and the edges of $B_1$. Let $R'=(R\setminus \{uv\})\cup R_2$ and $B'=(B\setminus B_1)\cup B_2$.
Then the instances $(G,R,B)$ and $(G',R',B')$ are equivalent: 
a graph $H$ is readily seen to be a solution for $(G,R,B)$ if and only if  $H-uv$ is a solution for $(G',R',B')$.
This completes the correctness proof of our algorithm.

It remains to evaluate the running time. We can find a recognizable edge $uv$ together with the corresponding $(u,v)$-partition $(X,Y)$ in time $O(mn^2)$. This can be seen as follows. For each edge $uv$, we find $Z=N_G(u)\cap N_G(v)$. 
Then we check conditions a) and b) of Definition~\ref{d-rec}, that is, we check whether $Z$ is the union of two disjoint cliques with no edges between them.
Finally, we check conditions~c) and~d) of 
Definition~\ref{d-rec}.
For a given $uv$, this can all be done in time $O(n^2)$. As we need to check at most $m$ edges, one iteration takes time $O(mn^2)$.  As the total number of iterations is at most $m$, the whole algorithm runs in time $O(n^2m^2)$.\qed
\end{proof}

\section{The Linear Kernel}\label{s-kernel}

For proving that {\sc Square Root with Labels} restricted to planar$+kv$ graphs has a linear kernel when parameterized by~$k$, 
we will use the following result of Harary, Karp and Tutte as a lemma.

\begin{lemma}[\cite{HKT67}]\label{lem:planar}
A graph $H$ has a planar square if and only if
\begin{itemize}
\item[i)]  every vertex $v\in V_H$ has degree at most~$3$,
\item[ii)] every block of $H$ with more than four vertices is a cycle of even length, and
\item[iii)] $H$ has no three mutually adjacent cut vertices. 
\end{itemize}
\end{lemma}

We need the following additional terminology.
A block is \emph{trivial} if it has exactly one vertex; note that this vertex must have degree~0.
A block  is \emph{small} if it has exactly two vertices and \emph{big} otherwise. 
We say that a block is \emph{pendant} if it is a small block with a vertex of degree~1.

We need two more structural lemmas.
We first show the effect of applying our {\bf Edge Reduction Rule} on the number of vertices in a connected component of a planar graph.

\begin{lemma}\label{lem:plan-red}
Let $G$ be a planar graph with a square root. If $G$ has no recognizable edges, then every connected component of $G$ has at most $12$ vertices.
\end{lemma}

\begin{proof}
Let $G$ be a planar square with no recognizable edges.
We may assume without loss of generality that $G$ is connected and $|V_G|\geq 2$. Let $H$ be a square root of $G$. Recall that $H$ is a connected spanning subgraph of $G$. Hence, it suffices to prove that $H$ has at most 12 vertices.

First suppose that $H$ does not have a big block, in which case every edge of~$H$ is a bridge. 
As $G$ has no recognizable edges, Corollary~\ref{lem:edge-one2} implies that every block of $H$ is pendant. By Lemma~\ref{lem:planar}, 
every vertex of $H$ degree at most~3. Hence, $H$ has at most four vertices.

Now suppose that $H$ has a big block $F$.
If $F$ contains no cut vertices of $H$, then $H=F$ has at most six vertices due to Corollary~\ref{lem:edge-one2} and Lemma~\ref{lem:planar}.
Assume that $F$ contains a cut vertex~$v$ of $H$.
Lemma~\ref{lem:planar} tells us that $d_H(v)\leq 3$; therefore $v$ is a vertex of exactly two blocks, 
namely $F$ and some other block~$S$. Because $F$ is big, $v$ has two neighbours in $F$. Hence, $v$ can only have one neighbour in $S$, thus
$S$ is small. As $G$ has no 
recognizable edges, Corollary~\ref{lem:edge-one2} implies that $S$ is a pendant block.
Hence, we find that $|V_G|\leq 2|V_F|$ (with equality if and only if each vertex of $F$ is a cut vertex).

If $F$ has at least seven vertices, then it follows from Lemma~\ref{lem:planar} that $F$ is a cycle of even length at least~8, which is not possible due to Corollary~\ref{lem:edge-one2}.
We conclude that $|V_F|\leq 6$ and find that $|V_G|=|V_H|\leq 2|V_F|\leq 12$.\qed
\end{proof}

We now prove our second structural lemma.

\begin{figure}[ht]
\centering\scalebox{0.9}{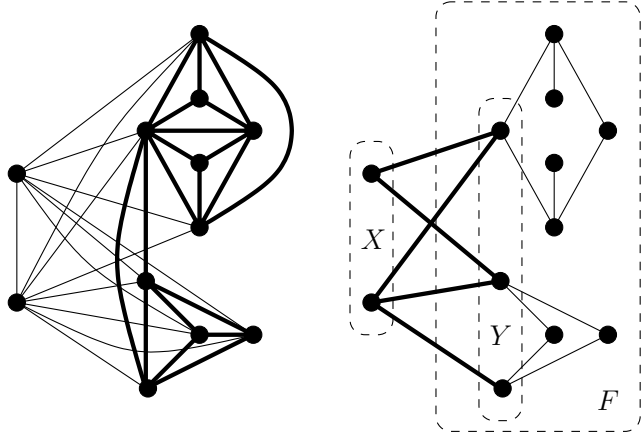}
\caption{An example of a planar$+2v$ graph $G=H^2$ (left side) and a square root $H$ of $G$ (right side). The thick edges in $G$ denote the planar component; the thick edges in $H$ denote the edges of~$A$.\label{f-2kv}}
\end{figure}

\begin{lemma}\label{lem:kern-size}
Let $G$ be a planar$+kv$ graph with no recognizable edges, such that every connected component of $G$ has at least $13$ vertices.
If $G$ has a square root, then $|V_G|\leq 137k$.
\end{lemma}

\begin{proof}
Let $H$ be a square root of $G$. 
By Lemma~\ref{lem:plan-red}, $G$ cannot have any planar connected components (as these would have at most~12 vertices).
Hence, every connected component of $G$ is non-planar.

Since $G$ is planar$+kv$, there exists a subset $X\subseteq V_G$ of size at most $k$ such that $G-X$ is planar.
Let $F=H-X$. Note that $F$ is a spanning subgraph of $G-X$ and that $F^2$ is a (spanning) subgraph of $G-X$; hence 
$F^2$ is planar.
Let $Y$ be the set that consists of all those vertices of $F$ that are a neighbour of $X$ in $H$, that is 
$Y=N_H(X)\cap V_F$.
Since every connected component of 
$G$ 
is non-planar, every connected component of $F$ contains at least one vertex of $Y$. 
Let $A$ be the set that consists of all edges between $X$ and $Y$ in $H$, that is, $A=\{uv\in E(H)\mid u\in X,v\in Y\}$.
See Figure~\ref{f-2kv} for an example.

Consider a vertex $v\in X$.
By Kuratowski's Theorem, the (planar) graph $G-X$ has no clique of size~5.
Since $N_H(v)\cap (V_G\setminus X)$ is a clique in $G-X$, we find that
 $|N_H(v)\cap (V_G\setminus X)|\leq 4$.
Hence, $|Y|\leq 4|X|\leq 4k$.

We now prove three claims about the structure of blocks of $F$.

\medskip
\noindent
{\bf Claim A.} {\it If $R$ is a block of $F$ that is not a pendant block of $H$, then $V_R$ is at distance at most $1$ from $Y$ in $F$.}

\medskip
\noindent
We prove Claim~A as follows. Let $R$ be a block of $F$ that is not a pendant block of $H$.
To obtain a contradiction, assume that $V_R$ is at distance at least~2 from $Y$ in $F$. Let $u$ be a vertex of $R$ such that $\dist_F(u,Y)=\min\{\dist_F(u,v)\mid v\in V_R\}$, so $u$ is a cut vertex of $F$ that is of distance at least~2 from $Y$ in $F$.
Note that $R$ is not a trivial block of $F$, since all trivial blocks are isolated vertices of $F$ that are vertices of $Y$.

First suppose that $R$ is a small block of $F$ and let $v$ be the second vertex of $R$.
Then the edge $uv$ is a bridge of $F$.
Since $R$ is not pendant, it follows from Corollary~\ref{lem:edge-one2} that $uv$ is in a cycle of length~$C$ at most~6 in $H$.
Observe that $C$ must contain at least two edges of $A$, which implies that $u$ or $v$ is at distance at most~1 from $Y$. This is 
a contradiction.

Now suppose that $R$ is a big block of $G$. Let $v$ be the neighbour of $u$ in a shortest path between $u$ and $Y$ in $F$.
By Lemma~\ref{lem:planar}, $u$ has degree at most~3 in $F$. As $R$ is big, $u$ has at least two neighbours in $F$. Hence, $uv$ is a bridge of $F$. 
As $v$ has at least two neighbours in $F$ as well, $uv$ is not a pendant edge of $H$.
Then it follows from Corollary~\ref{lem:edge-one2} that $uv$ is in a cycle~$C$ of length at most~6 in $H$.
Observe that $C$ must contain at least two edges of $A$ and at least one edge $uw$ of $R$ for some vertex $w\neq u$ in $R$.
Hence, $w$ is at distance at most 1 from $Y$, which is a contradiction.
This completes the proof of Claim~A.

\medskip
\noindent
By Lemma~\ref{lem:planar}, every vertex of $F$ has degree at most~3 in~$F$. 
Hence the following holds:

\medskip
\noindent
{\bf Claim B.} {\it For every $u\in Y$, $F$ has at most three big blocks at distance at most~$1$ from $u$.}

\medskip
\noindent
Let $Z$ be the set of vertices of $F$ at distance at most 3 from $X$ in $H$.

\medskip
\noindent
{\bf Claim C.} {\it If $R$ is a block of $F$ with $V_R\setminus Z\neq\emptyset $, then $|V_R|\leq 6$.}

\medskip
\noindent
We prove Claim~$C$ as follows. Suppose $R$ is a block of $F$ with $V_R\setminus Z\neq\emptyset$. 
For contradiction, assume that $|V_R|\geq 7$.
Then, by Lemma~\ref{lem:planar}, $R$ is a cycle of $F$ of even size.
As $V_R\setminus Z\neq\emptyset$ and $R$ is connected, there exists an edge $uv$ of $F$ with $u\notin Z$.
By Corollary~\ref{lem:edge-one2}, we find that
$uv$ is in a cycle~$C$ of $H$ of length at most~6.
Since $u$ is at distance at least~4 from $X$ in $H$, we find that $C$ contains no vertex of $X$ and therefore, $C$ is a cycle of $F$.
Then $R=C$ must hold, which is a contradiction as $|V_R|\geq 7>6\geq |V_C|$.
This completes the proof of Claim~$C$.

\medskip
\noindent
We will now show that the diameter of $F$ is bounded.
We start with proving the following claim.

\medskip
\noindent
{\bf Claim D.} {\it Every vertex of every block $R$ of $F$ that is non-pendant in $H$ is at distance at most~$5$ from~$X$ in $H$. Moreover,
\begin{itemize}
\item[i)]  if $R$ has a vertex at distance at least $4$ from $X$ in $H$, then $R$ is a big block,
\item[ii)] $R$ has at most three vertices at distance at least $4$  and at most one vertex at distance~$5$ from $X$ in $H$.
\end{itemize}
}

\medskip
\noindent
We prove Claim~D as follows.
Let $R$ be a block of $F$ that is non-pendant in $H$.  
Claim~A tells us that $V_R$ is at distance at most 1 from $Y$ in $F$.

If $R$ is a small block, then every vertex of $R$ is at distance at most 2 from $Y$. Hence, every vertex of $R$ is at distance at most 3 from $X$ in $H$ and the claim holds for $R$. 

Let $R$ be a big block.
If $R$ has at most four vertices, then the vertices of $R$ are at distance at most 3 from $Y$ in $F$ and at most one vertex of $R$ is at distance
exactly~3. Hence, the vertices of $R$ are at distance at most 4 from $X$ in $H$ and at most one vertex of $R$ is at distance
exactly~4.
Assume that $|V_R|>4$. 
Then either $V_R\subseteq Z$, that is, all the vertices are at distance at most 3 from $X$ in $H$, 
or,
by Lemma~\ref{lem:planar} and Claim~C, we find that $R$ has at most six vertices.
As $|V_R|>4$, we find that $R$ is a cycle on six vertices by Lemma~\ref{lem:planar}.
Hence, in the latter case every vertex of $R$ is at distance at most~4 from $Y$, that is, at distance at most~5 from $X$ in $H$. Moreover,  at most three vertices are at distance at least~4 and at most one vertex is at distance 5 from $X$ in $H$ as $R$ is a cycle. 
This completes the proof of Claim~D.

\medskip
\noindent
By combining Claim~B with the fact that $|Y|\leq 4k$, we find that $F$ has at most $12k$ big blocks at distance at most~1 from $Y$. By Claims~A and D, this implies that $H$ has at most $36k$ vertices of non-pendant blocks at distance at least 4 from $X$ in $H$ and at most $12k$ vertices at distance at least~5 from $X$ in $H$.
Let $v$ be a vertex $H$ of degree 1 in $H$.
If $v$ is at distance at least 5 from $X$, then $v$ is adjacent to a vertex $u$ of a non-pendant block and $u$ is at distance at least 4 from $X$ in $H$.
Notice that $v$ is a unique vertex of degree 1 adjacent to $u$, because by Claim~D, $u$ is in a big block and $d_F(u)\leq 3$ by Lemma~\ref{lem:planar}. 
Since $H$ has at most $36k$ vertices of non-pendant blocks at distance at least 4 from $X$ in $H$, the total number of vertices of degree 1 at distance at least 5 from $X$ in $H$ is at most $36k$. 
Taking into account that there are at most $12k$ vertices at distance at least 5 from $X$ in $H$ in non-pendant blocks, we see that there are at most $48k$ vertices in $H$ at distance at least 5 from $X$ and all other vertices in $F$ are at distance at most 4 from $X$.
Using the facts that $|Y|\leq 4k$ and that $d_F(v)\leq 3$ for $v\in V_F$ by Lemma~\ref{lem:planar}, we observe that $H$ has at most 
$k+4k+12k+24k+48k=89k$ vertices at distance at most 4 from $X$.
It then follows that $|V_G|=|V_H|\leq 48k+89k=137k$.\qed 
\end{proof}

We are now ready to prove our main result.

\begin{theorem}\label{thm:kern}
\textsc{Square Root with Labels} has a kernel of size $O(k)$ for planar$+kv$ graphs when parameterized by $k$.
\end{theorem}

\begin{proof}
Let $(G,R,B)$ be an instance of \textsc{Square Root with Labels}.
First we apply {\bf Edge Reduction}, which takes polynomial time due to Lemma~\ref{lem:preproc}.
By the same lemma we either solve the problem in polynomial time or obtain an equivalent instance $(G',R',B')$ with the following property: for any square root $H$ of $G'$, every edge of $H$ is either a pendant edge of $H$ or is included in a cycle of length at most~6 in $H$.
In the latter case we apply the following reduction rule exhaustively, which takes polynomial time as well.

\medskip
\noindent
{\bf Component Reduction.} If $G'$ has a connected component $F$ with $|V_F|\leq 12$, then use brute force to solve \textsc{Square Root with Labels} for 
$(F,R'\cap V_F,B'\cap V_F)$.
If this yields a {\tt no}-answer, then return {\tt no} and stop.
Otherwise, return $(G'-V_F,R'\setminus V_F,B'\setminus V_F)$ or if $G'=F$, return {\tt yes} and stop.

\medskip
\noindent
It is readily seen that this rule either solves the problem correctly or returns an equivalent instance. 
Assume we obtain an instance $(G'',R'',B'')$.
Our reduction rules do not increase the deletion distance, that is, $G''$ is a planar$+kv$ graph.
Then by Lemma~\ref{lem:kern-size}, if $G''$ has more than $137k$ vertices then $G''$, and thus $G$, has no square root.
Hence, if $|V_G''|>137k$, we have a no-instance,
in which case we return a {\tt no}-answer and stop.
Otherwise, we return the kernel $(G'',R'',B'')$.\qed
\end{proof}

\section{Another Application}\label{sec:mad}

In this section we give another application of the {\bf Edge Reduction} rule.
Let $G$ be a graph. The \emph{average degree} of $G$ is $\ad(G)=\frac{1}{|V_G|}\sum_{v\in V_G}d_G(v)=\frac{2|E_G|}{|V_G|}$.  
Then the \emph{maximum average degree} of $G$ is defined as $$\mad(G)=\max\{\ad(H)\; |\; H\text{ is a subgraph of }G\}.$$
We will show that {\sc Square Root} is polynomial-time solvable for graphs with maximum average degree less than 
$\frac{46}{11}$.  

In order to prove our result we will need a number of lemmas, amongst others three lemmas on treewidth.
A \emph{tree decomposition} of a graph $G$ is a pair $(T,X)$ where $T$
is a tree and $X=\{X_{i} \mid i\in V_T\}$ is a collection of subsets (called {\em bags})
of $V_G$ such that the following three conditions hold: 
\begin{itemize}
\item[i)] $\bigcup_{i \in V_T} X_{i} = V_G$, 
\item[ii)] for each edge $xy \in E_G$, $x,y\in X_i$ for some  $i\in V_T$, and 
\item[iii)] for each $x\in V_G$ the set $\{ i \mid x \in X_{i} \}$ induces a connected subtree of $T$.
\end{itemize}
The \emph{width} of a tree decomposition $(\{ X_{i} \mid i \in V_T \},T)$ is $\max_{i \in V_T}\,\{|X_{i}| - 1\}$. The \emph{treewidth} $\tw(G)$ of a graph $G$ is the minimum width over all tree decompositions of $G$. 
A class of graphs~${\cal G}$ has {\it bounded treewidth} if there exists a constant $p$ such that the treewidth of every graph from ${\cal G}$ is at most~$p$.

The first lemma is known
and shows that {\sc Square Root with Labels} is linear-time solvable for graphs of bounded treewidth. We give a proof for completeness.

\begin{lemma}[\cite{CochefertCGKP13}]\label{lem:tw}
The {\sc Square Root with Labels} problem can be solved in  
time $O(f(t)n)$ for $n$-vertex graphs of treewidth at most $t$.
\end{lemma}

\begin{proof}
It is not difficult to construct a dynamic programming algorithm for the problem, but for simplicity, we give a non-constructive proof based on Courcelle's theorem~\cite{Courcelle92}. By this theorem, it suffices to show that the existence of a square root~$H$ of a graph~$G$ can be expressed in monadic second-order logic. To see the latter, note that the existence of  a graph $H$ with $H^2=G$, $R\subseteq E_H$  and $B\cap E_H=\emptyset$ is equivalent to the existence of an edge subset $X\subseteq E_G$ that satisfies the following conditions:
\begin{itemize}
\item [(i)] $R\subseteq X$
\item [(ii)] $B\cap X=\emptyset$
\item [(iii)] for any $uv\in E_G$, either $uv\in X$ or there is a vertex $w$ with $uw,wv\in X$
\item [(iv)] for any two distinct edges $uw,wv\in X$ we have $uv\in E_G$.
\end{itemize}
This completes the proof of the lemma.\qed
\end{proof}

The second lemma is a well-known result about deciding whether a graph has treewidth at most $k$ for some constant $k$.

\begin{lemma}[\cite{Bodlaender96}]\label{l-bod}
For any fixed constant~$k$, it is possible to decide in linear time whether the treewidth of a graph is at most $k$.
\end{lemma}

We also need the following result as a third lemma.

\begin{lemma}\label{l-first}
Let $H$ be a square root of a graph $G$.
Let $T$ be the bipartite graph with $V_T=\mathcal{C}\cup\mathcal{B}$, where partition classes $\cal C$ and $\cal B$ are 
the set of cut vertices and blocks of $H$, respectively, such that 
$u\in \mathcal{C}$ and $Q\in\mathcal{B}$ are adjacent if and only if $Q$ contains $u$.
For $u\in\mathcal{C}$, let $X_u$ consist of $u$ and all neighbours of $u$ in $H$.
For $Q\in \mathcal{B}$, let $X_Q=V_Q$.
Then $(T,X)$ is a tree decomposition of $G$.
\end{lemma}

\begin{proof}
We first prove that $T$ is a tree. For contradiction, suppose that $T$ contains a cycle. Then this cycle is of the form $Q_1u_1\cdots Q_pu_pQ_1$ for some integer $p\geq 2$, 
where $Q_1,\ldots,Q_p$ are blocks of $H$ and $u_1,\ldots,u_p$ are cut vertices of $H$.
By using this cycle we find that $H$ has 
one path from $u_1$ to $u_2$ that is contained in $H[Q_1]$ and one path from $u_1$ to $u_2$ that is contained in $H[Q_2\cup \cdots \cup Q_p]$.
This contradicts our assumption that $u_1$ and $u_2$ are cut vertices of $H$.

We now prove that $(T,X)$ satisfies the three conditions (i)--(iii) of the definition of a tree decomposition.
Condition (i) is satisfied, as every vertex of $H$, and thus every vertex of $G$,  belongs to some block $Q$ of $H$ and thus to some bag $X_Q$.
Condition (ii) is satisfied, as every two 
vertices $x,y$ that are adjacent in $G$ either belong to some common block $Q$ of $H$, and thus belong to $X_Q$, or else have a common neighbour $u$ in $H$ that is a cut vertex of $H$, and thus belong to $X_u$.

In order to prove (iii), consider a vertex $x\in V_G$. 
First suppose that $x$ is a cut vertex of $H$. Then the set of bags to which $x$ belongs consists of bags $X_Q$ for every block $Q$ of $H$ to which $x$ belongs and bags $X_u$ for every neighbour $u$ of $x$ in $H$ that is a cut vertex of $H$. Note that $x$ and any neighbour $u$ of $x$ in $H$ belong to some common block of $H$. Hence, by definition, the corresponding nodes in $T$ form a connected induced subtree of $T$ (which is a star in which every edge is subdivided at most once).
Now suppose that $x$ is not a cut vertex of $H$. Then $x$ is contained in exactly one block $Q$ of $H$.
Hence the set of bags to which $x$ belongs consists of the bags $X_Q$ and bags $X_u$ for every neighbour $u$ of $x$ in $H$ that is a cut vertex of $H$. Note that such a neighbour $u$ belongs to $Q$. Hence, by definition, the corresponding nodes in $T$ form a connected induced subtree of $T$ (which is a star).
This completes the proof of Lemma~\ref{l-first}.\qed
\end{proof}
We call the tree decomposition $(T,X)$ of Lemma~\ref{l-first} the {\it $H$-tree decomposition} of~$G$.
Finally, the fourth lemma shows why we need the previous lemmas.

\begin{lemma}\label{lem:mad-tw}
Let $G$ be a graph with $\mad(G)<\frac{46}{11}$. If $G$ has a square root but no recognizable edges, then $\tw(G)\leq 5$.
\end{lemma}

\begin{proof}
We assume without loss of generality that $G$ is connected; otherwise we can consider the connected components of $G$ separately. We also assume that $G$ has at least one edge, as otherwise the claim is trivial. 
 Let $H$ be a square root of $G$.
Let $\cal C$ be the set of cut vertices of $H$, and let $\cal B$ be the set of blocks of $H$.
We construct the $H$-tree decomposition $(T,X)$ of $G$ (cf. Lemma~\ref{l-first}).
We will show that $(T,X)$ has width at most 5.

If $v\in V_H$, then $N_H[v]$ is a clique in $G$. Hence,  $\Delta(H)\leq 4$ because
otherwise $\ad(G[N_H[v]])\geq 5$, contradicting our assumption that $\mad(G)<\frac{46}{11}$.
Hence each bag $X_u$ corresponding to a cut vertex $u$ of $H$ has size at most~4. 
We claim that each bag corresponding to  a block of $H$ has size at most~6, that is, we will prove that each block of $H$ has at most 
six vertices. For contradiction, assume that $|V_Q|\geq 7$ for some block $Q$ of $H$.

First assume that $Q$ is a cycle. Then, as $|V_Q|\geq 7$, no edge of $Q$ is included in a cycle of length at most~6. 
Then by Lemma~\ref{lem:edge-one} we find that $Q$ and, consequently, $H$ has a recognizable edge, contradicting our assumption that $H$ does not have such edges. Hence, $Q$ contains at least one vertex that does not have degree~2 in $Q$. As $Q$ is 2-connected and  a 2-connected graph not isomorphic to $K_2$, we find that $Q$ has no pendant vertices. Hence, $Q$ contains at least one vertex of degree at least~3 in $Q$.

We claim that for any vertex $u\in V_Q$, $d_{Q^2}(u)\geq d_{Q}(u)+2$. In order to see this, let $S=\{v\in V_Q\; |\; \dist_Q(u,v)=2\}$.  Because 
$Q$ is connected, $d_Q(v)\leq 4$ 
and
$|V_Q|\geq 7$, we find that $S\neq\emptyset$. 
If $S=\{v\}$ for some $v\in V_Q$, then $v$ is a cut vertex of $Q$, contradicting the 2-connectedness of $Q$. Therefore, $|S|\geq 2$ and thus $d_{Q^2}(u)\geq d_Q(u)+|S|\geq d_Q(u)+2$. 

We also need the following property of $Q$. Let $u,v$ be two distinct vertices of degree at least 3 in $Q$ joined by a path $P$ in $Q$ of length 5 such that all inner vertices of $P$ have  degree 2 in $Q$. We claim that in any such case $u$ and $v$ are not adjacent in $Q$. 
In order to see this, assume that $uv\in E_Q$. Let $x$ and $y$ be neighbours of $u$ and $v$, respectively, that are not in $P$. If $x=y$, then $\ad(Q[V_P\cup\{x\}]^2)=\frac{32}{7}\geq\frac{46}{11}$. If $x\neq y$,   
then $\ad(Q[V_P\cup\{x,y\}]^2)\geq\frac{36}{8}\geq\frac{46}{11}$. In both cases we get a contradiction with our assumption that $\mad(G)<\frac{46}{11}$. Hence, $uv\notin E_Q$. 

We use the property deduced above as follows.
Consider any two distinct vertices $u$ and $v$ of degree at least 3 in $Q$ that are
joined by a path $P$ in $Q$ such that all inner vertices of $P$ have degree 2 in $Q$. Then the length of $P$ is at most 4. Otherwise, because $uv\notin E_Q$, no edge of $P$ is included in a cycle of length at most 6, and then by Lemma~\ref{lem:edge-one} we find that $Q$ and, consequently, $H$ has a recognizable edge, contradicting our assumption that $H$ does not have such edges.

Recall that every vertex in $Q$ has degree between 2 and 4 in $Q$.
We let $p,q$ and $r$ be the numbers of vertices of $Q$ of degree~2,~3 and~4, respectively, in~$Q$. We  construct an auxiliary multigraph $F$ as follows. The vertices of $F$ are the vertices of $Q$ of degree 3 and 4. For any path $P$ in $Q$ between two vertices $u$ and $v$ of $F$ with 
the property that  all inner vertices of $P$ have degree 2 in $Q$, we add an edge $uv$ to $F$. Note that $P$ may have length~1. We also note that $F$ can have multiple edges but no self-loops, because $Q$ is 2-connected. 
Moreover, we observe that $F$ has $q+r$ vertices and $\frac{1}{2}(3q+4r)$ edges. As each path in $Q$ that corresponds to an edge of $F$ has length at most 4, we find that $p\leq \frac{3}{2}(3q+4r)$. Recall that $Q$ has at least one vertex with degree at least~3 in $Q$; hence, $\max\{q,r\}\geq 1$. This means that
$$\ad(Q)=\frac{2p+3q+4r}{p+q+r}=\frac{q+2r}{p+q+r}+2\geq \frac{2q+4r}{11q+14r}+2\geq\frac{2}{11}+2.$$
Because $d_{Q^2}(u)\geq d_Q(u)+2$ for each $u\in V_Q$, the above inequality implies that $\ad(Q^2)\geq\frac{2}{11}+4=\frac{46}{11}$; a contradiction. Hence, $H$ cannot have blocks of size at least~7.\qed
\end{proof}

We are now ready to prove the main result of this section.

\begin{theorem}\label{thm:mad}
 {\sc Square Root} can be solved in time $O(n^4)$ for $n$-vertex  graphs $G$ with $\mad(G)<\frac{46}{11}$.  
\end{theorem}

\begin{proof}
Let $G$ be an $n$-vertex graph with $\mad(G)<\frac{46}{11}$.
Our algorithm consists of the following two stages:

\medskip
\noindent
{\it Stage 1.} We construct the instance $(G,R,B)$ of  {\sc Square Root with Labels}  from 
$G$ by setting $R=B=\emptyset$. Then we preprocess $(G,R,B)$ using {\bf Edge Reduction}. By
Lemma~\ref{lem:preproc}, we either solve the problem (and answer {\sc no}) or obtain an equivalent instance $(G',R',B')$ of  {\sc Square Root with Labels} that has no recognizable edges by Lemma~\ref{lem:edge-one}. 
If we get an instance $(G',R',B')$ then we proceed with the second stage.

\medskip
\noindent
{\it Stage 2.} We solve instance $(G',R',B')$ as follows.
By Lemma~\ref{lem:mad-tw}, if $G'$ has a square root, then $\tw(G')\leq 5$. We check the latter property by using Lemma~\ref{l-bod}.
If $\tw(G')\geq 6$, then we stop and return {\tt no}. Otherwise, we solve the problem by Lemma~\ref{lem:tw}.

\medskip
\noindent
By Lemma~\ref{lem:preproc}, stage 1 takes time $O(n^2m^2)$, where $m$ is the number of edges of $G$. Since $\mad(G)<\frac{46}{11}$, stage~1 runs in fact in time $O(n^4)$. As stage 2 takes $O(n)$ time by Lemmas~\ref{lem:tw} and~\ref{l-bod}, the total running time is $O(n^4)$.\qed
\end{proof}

\section{Conclusions}\label{s-con}

We proved a linear kernel for 
{\sc Square Root with Labels}, which generalizes the {\sc Square Root} problem,
for planar$+kv$ graphs using a new edge reduction rule. 
We recall that our edge reduction rule can be applied to solve {\sc Square Root} for graphs of maximum degree at most~6~\cite{CCGKPS}.
To illustrate its wider applicability we gave a third example of our edge reduction rule by showing that it can  
be used to solve {\sc Square Root} in polynomial time for graphs with maximum average degree less than 
$\frac{46}{11}$.  
Whether {\sc Square Root} is polynomial-time solvable for graphs of higher maximum average degree or for graphs of maximum degree at most~7 is still open. In general, it would be interesting to research whether our edge reduction rule can be used to obtain other polynomial-time results for {\sc Square Root}.

\end{document}